\documentclass{LMCS}

\def\doi{9(4:7)2013}
\lmcsheading%
{\doi}
{1--10}
{}
{}
{Mar.~\phantom08, 2013}
{Oct.~16, 2013}
{}

\usepackage{enumerate}
\usepackage{hyperref}
\usepackage{amssymb}
\usepackage{bussproofs}
\usepackage{tikz-qtree}

\theoremstyle{plain}

\begin{document}

\title[Implicit Resolution]{Implicit Resolution}

\author[Z.~C.~Wang]{Zi Chao Wang}	%required
\address{Department of Algebra\\Faculty of Mathematics and Physics\\Charles University in Prague\\Sokolovsk\' a 83\\Praha 8-Karl\' in\\18675\\Czech Republic}	%required
\email{wang@karlin.mff.cuni.cz}  %optional
\thanks{The research leading to these results has received funding from the European Community's Seventh Framework Programme FP7/2007-2013 under grant agreement n\ensuremath{^\circ} 238381.}	%optional

%% etc.

%% required for running head on odd and even pages, use suitable
%% abbreviations in case of long titles and many authors:

%% mandatory lists of keywords and classifications:
\subjclass{MANDATORY list of acm classifications}
\ACMCCS{[{\bf Theory of computiation}]:  Computational complexity and cryptography---Proof complexity}
%\titlecomment{OPTIONAL comment concerning the title, \eg, if a variant
%or an extended abstract of the paper has appeared elsewehere}
%%%%%%%%%%%%%%%%%%%%%%%%%%%%%%%%%%%%%%%%%%%%%%%%%%%%%%%%%%%%%%%%%%%%%%%%%%%

%% the abstract has to PRECEED the command \maketitle:
%% be sure not to issue the \maketitle command twice!

\begin{abstract}
  Let $\Omega$ be a set of unsatisfiable clauses, an implicit
  resolution refutation of $\Omega$ is a circuit $\beta$ with a
  resolution proof $\alpha$ of the statement ``$\beta$ describes a
  correct tree-like resolution refutation of $\Omega$''. We show that
  such a system is p-equivalent to Extended Frege.

  More generally, let $\tau$ be a tautology, a $[P, Q]$-proof of
  $\tau$ is a pair $(\alpha, \beta)$ s.t. $\alpha$ is a $P$-proof of
  the statement ``$\beta$ is a circuit describing a correct $Q$-proof
  of $\tau$''. We prove that $[EF,P] \leq_p [R,P]$ for an arbitrary
  Cook-Reckhow proof system $P$.
\end{abstract}

\maketitle

%% start the paper here:

\section{Introduction}

In proof complexity one of the basic questions that remained open is whether or not there is an optimal proof system. Although there is no consensus whether such proof system should exist it is generally believed that Extended Frege is the pivotal case in the sense that if such optimal proof system exists then Extended Frege is currently the most natural candidate. This is because Extended Frege corresponds to the complexity class $P/poly$ and many attempts in constructing proof systems that are conjecturally stronger than Extended Frege ended up in producing systems that are equivalent to Extended Frege.

Implicit proofs were introduced by Kraj{\'{\i}}{\v{c}}ek \cite{MR2058178} as a general framework for direct combinatorial constructions of strong proof systems beyond Extended Frege. The idea is to succinctly describe an exponential size proof by some polynomial size circuit then supplement such circuit with an additional correctness proof.
Loosely speaking, let $P$ and $Q$ be some existing proof systems and $\tau$ be a tautology, a $[P, Q]$-proof of $\tau$ is a pair $(\alpha, \beta)$ s.t. $\alpha$ is a $P$-proof of the formalized statement ``$\beta$ is a circuit describing a correct $Q$-proof of $\tau$''. For any proof system $P$ the implicit version of $P$, denoted $iP$, is the proof system defined as $[P,P]$.

Whilst a hierarchy of implicit proof systems based on Extended Frege were introduced in \cite{MR2058178}, the system $iEF$, is of specific interest. We may think of $iEF$ as the ``succinct'' version of exponential size Extended Frege proofs and it bears a correspondence to exponential time computation and serves as the base case of the iterated construction of a strong implicit proof system whose soundness is not provable in the theory $T^2+Exp$ where the exponentiation function is total. We would therefore expect that insights into the problem whether $iEF$ is indeed stronger than Extended Frege, shall contribute to the empirical evidences towards the study of the existence of optimal proof systems.

In contrast to strong proof systems such as Extended Frege of which we do not even have any candidate hard tautologies, resolution is a refutational proof system with the resolution rule as the only derivation rule. It had been extensively studied since its introduction and substantial progress had been made in understanding its limits. Resolution is known to be inefficient for proving a number of combinatorial principles. For example, Haken \cite{MR821207} first proved that the propositional pigeon hole principle requires exponential size resolution refutations. More recently, systematic treatment on lower bounds of resolution in terms of clause width was presented by Ben-Sasson and Widgerson in \cite{MR1868713}.

In this paper we are motivated to understand Extended Frege in terms of resolution and implicit proofs. In Theorem \ref{t_main} we show that Extended Frege is p-equivalent to a resolution based proof system in the framework of implicit proofs\footnote{This was first conjectured by Kraj{\'{\i}}{\v{c}}ek.}. We generalize the construction in Theorem \ref{t_gen} to prove that $[EF,P] \leq_p [R,P]$ for any proof system $P$, hence showing that $iEF$ collapses to $[R,EF]$, although we are not able to address the precise strength of the latter. As a by product, in Lemma \ref{t_c} we show that existence of an $NP$ search algorithm that is provably correct in Extended Frege implies existence of such algorithm provably correct in resolution.

The paper is organized as follows. We briefly review the definition of resolution and fix notation in Section \ref{s_pre}. In Section \ref{s_circuit} we present a prototype of the key technical construction in terms of correctness of $NP$ search algorithms. In Section \ref{s_ires} we give precise definition of implicit resolution and prove the main result that it is p-equivalent to Extended Frege. In Section \ref{s_gen} we outline the construction applied to general implicit proof systems and briefly discuss generalizations to subsystems of $EF$.

\section{Preliminaries}
\label{s_pre}

Recall that a propositional proof system $P$, as defined by Cook and Reckhow in \cite{MR523487}, is a polynomial time Turing machine $P$ s.t. for any propositional formula $\tau$ there exists a string $\pi$ with $P(\pi) = \tau$ if and only if $\tau$ is a tautology, in which case we say that $\pi$ is a $P$-proof of $\tau$. Let $P$ and $Q$ be proof systems. We say that $P$ simulates $Q$, denoted $Q \leq P$, if there exists some constant $c$ s.t. for any tautology $\tau$ with $Q$-proof of length $n$ there exists a $P$-proof of $\tau$ of length less than $n^c$. We also say that $P$ p-simulates $Q$, denoted $Q \leq_p P$, if we are able to compute the required $P$-proof from the given $Q$-proof in time $n^c$. We say that a proof system $P$ is optimal if $Q \leq P$ for any proof system $Q$ \cite{MR1011192}.

A literal is a propositional variable or the negation of a propositional variable, i.e., we say that $l$ is a literal if $l := p$ or $l := \overline{p}$ for some propositional variable $p$. We also write $l^1:=l$ to denote $l$ and $l^0:=\overline{l}$ to denote the negation of l. A clause is a multiset of literals, which we shall interpret as a disjunction. Let $L$ be a clause we would also write $l \in L$ for some literal $l$ if $L := l \lor l_1 \lor ... \lor l_n$ for literals $l_1,...,l_n$. We interpret a set of clauses as a conjunction of disjunctions of literals.

The resolution rule and the weakening rule are the following:
\begin{prooftree}
\AxiomC{$A \lor x$}
\AxiomC{$\overline{x} \lor B$}
\LeftLabel{r}
\BinaryInfC{$A \lor B$}
\DisplayProof\hspace{4cm}
\AxiomC{$C$}
\LeftLabel{w}
\UnaryInfC{$C \lor D$}
\end{prooftree}
where $A$, $B$, $C$ and $D$ are clauses, $x$ is a propositional variable\footnote{The standard definition also requires that $\overline{x} \notin A$ and $x \notin B$.}. We say that the clause $A \lor B$ is derived from $A \lor x$ and $\overline{x} \lor B$ by resolving the variable $x$ and the clause $C \lor D$ is a weakening of the clause $C$, respectively.

Let $\Omega$ be a set of unsatisfiable clauses, a resolution refutation or $R$-refutation of $\Omega$ is a sequence of clauses $C_1,...,C_n$ s.t. $C_n$ is the empty clause and for each $i = 1,...,n$ either $C_i \in \Omega$ or $C_i$ is derived from $C_j$ and $C_k$ with the resolution rule for some $j<k<i$ or $C_i$ is a weakening of $C_j$ for some $j<i$.

We may think of an $R$-refutation as a directed acyclic graph and we say that an $R$-refutation is tree-like, denoted a $R^*$-refutation, if the directed acyclic graph is in fact a tree. We also say that an $R$-refutation is regular if on every path from a source to a sink in the directed acyclic graph no propositional variable is resolved more than once. Tseitin \cite{0205.00402} first introduced regular resolution and proved that regular resolution p-simulates tree-like resolution.

Resolution is complete without the weakening rule, i.e., let $\Omega$ be a set of clauses, if there exists an $R$-refutation of $\Omega$ with weakening then there exists an $R$-refutation of $\Omega$ without weakening.

\section{Resolution Refutation of Circuit Correctness}
\label{s_circuit}

In this paper we work with Extended Resolution, denoted $ER$, which is known to be p-equivalent to Extended Frege. In order to fix notation we give precise definition of Extended Resolution.

\begin{defi}
Let $C$ be a set of clauses we write $\Gamma(C)$ to denote the set of all propositional variables in $C$.
\end{defi}

Recall that an extension clause $e \equiv l_1 \lor l_2$ is an abbreviation for a group of clauses of the form $\overline{e} \lor l_1 \lor l_2, e \lor \overline{l_1}, e \lor \overline{l_2}$ where $e$ is a propositional variable and $l_1$ and $l_2$ are literals. Suppose $p_1$ and $p_2$ are the variables on which the literals $l_1$ and $l_2$ are defined, i.e., $l_1 := p_1^a$ and $l_2 := p_2^b$ for some $a,b \in \{0,1\}$ then we say that the extension clause $e \equiv l_1 \lor l_2$ defines two directed edges $p_1 \longrightarrow e$ and $p_2 \longrightarrow e$.

\begin{defi}
Let $C$ be a set of extension clauses. We say that $C$ is a circuit if the directed graph defined by the extension clauses is acyclic and every vertex has at most two incoming edges. We say that $v \in \Gamma(C)$ is a free variable if $v$ is a source, otherwise we say that $v$ is an extension variable. We write ${\mathcal F}(C)$ and ${\mathcal E}(C)$ to denote the set of free variables and the set of extension variables of the circuit $C$ respectively.
\end{defi}

\begin{defi}
Let $C$ be a set of unsatisfiable clauses, an $ER$-refutation of $C$ is an $R$-refutation of $C \cup D$ for some circuit $D$ with ${\mathcal F}(D) \subseteq \Gamma(C)$ and ${\mathcal E}(D) \cap \Gamma(C) = \emptyset$.
\end{defi}

For convenience we do not restrict the outputs of a circuit to the sinks of the graph and we shall instead specify explicitly the variables for which we would label as outputs. It is easy to see that such relaxed definition is equivalent to the conventional definition. Let $C$ be a circuit and let $v \in \Gamma(C)$ be a vertex which is not a sink be our desired output. We take $C' := C \cup \{ v' \equiv v \lor v \}$ for some $v' \notin \Gamma(C)$ then $v'$ would be a sink in the new circuit $C'$. Similarly suppose we have some sink $v \in C$ yet we do not label $v$ as an output we may take some output $u \in \Gamma(C)$ and define $C' := C \cup \{ v' \equiv v \lor \overline{v}, u' \equiv u \lor \overline{v'} \}$.

We shall write $C[x_1,...,x_m;v_1,..,v_n]$ to denote the circuit $C$ with variables $x_1,...,x_m \in {\mathcal F}(C)$ and $v_1,...,v_n \in {\mathcal E}(C)$ displayed. Note in particular that we do not necessarily display all variables. When we write $C \cup D$ for some circuits $C$ and $D$ we always mean their union as sets of clauses. Therefore the resulting set of clauses $C \cup D$ is not necessarily a circuit since conditions for extension variables might be violated.

\begin{defi}
Let $C$ and $D$ be circuits, we say that $C$ embeds in $D$ and $D$ is an expansion of $C$ if there exists an injection $f: \Gamma(C) \longrightarrow \Gamma(D)$ s.t. $f({\mathcal F}(C)) \subseteq {\mathcal F}(D)$ and $f(e) \equiv f(l_1) \lor f(l_2) \in D$ whenever $e \equiv l_1 \lor l_2 \in C$ where $f(l)=\overline{f(p)}$ if $l=\overline{p}$ for some variable $p$. We say that $f$ is an embedding from $C$ to $D$. We also say that $f$ is an isomorphism if $f$ is bijective.
\end{defi}

\begin{lem}
\label{l_eq}
Let $C$ and $D$ be circuits, if $D$ is an expansion of $C$ and $f: \Gamma(C) \longrightarrow \Gamma(D)$ is an embedding from $C$ to $D$ s.t $f(x)=x$ for all $x \in {\mathcal F}(C)$ then for any $y \in \Gamma(C)$ there exists $R$-refutations of $C \cup D \cup \{ y, \overline{f(y)} \}$ and $C \cup D \cup \{ \overline{y}, f(y) \}$ of size $O(|C|)$.
\end{lem}
\begin{proof}
By induction on the size of the circuit $C$. The base case is trivial so suppose we have $C' := C \cup \{ e \equiv l_1 \lor l_2 \}$ for some extension variable $e$ and literals $l_1$ and $l_2$.

From $C' \cup D \cup \{e, \overline{f(e)}\}$ we can derive the following clauses:

\begin{tabular}{ c c c }
\begin{tikzpicture}[grow'=up]
            \Tree [.$l_1\lor l_2$ $e$ $e\lor l_1\lor l_2$ ]
\end{tikzpicture}
&
\begin{tikzpicture}[grow'=up]
            \Tree [.$\overline{f(l_1)}$ $\overline{f(e)}$ $f(e)\lor\overline{f(l_1)}$ ]
\end{tikzpicture}
&
\begin{tikzpicture}[grow'=up]
            \Tree [.$\overline{f(l_2)}$ $\overline{f(e)}$ $f(e)\lor\overline{f(l_2)}$ ]
\end{tikzpicture}
\end{tabular}

\noindent By induction hypothesis there is an $R$-derivation of the singleton clause $\{l_2\}$ from $C \cup D \cup \{ l_1 \lor l_2, \overline{f(l_1)} \}$ of size $O(|C|)$. Then again by induction hypothesis there is an $R$-refutation of $C \cup D \cup \{ l_2, \overline{f(l_2)} \}$ of size $O(|C|)$.

The $R$-refutation of $C' \cup D \cup \{\overline{e}, f(e)\}$ is completely analogous.
\end{proof}

\noindent Suppose we have some circuit $C$ and such circuit appears as subcircuit of some larger circuit more than once, possibly with different inputs, as $D$ say, to avoid variable name collision it is mandatory for us to rename all the extension variables in $D$ and sometimes it is also necessary to rename the free variables. Therefore in order to speak about isomorphic subcircuits with potentially different input variables we define a duplicate of a circuit $C$ with variable substitutions $v_1/v'_1,...,v_n/v'_n$ to be an isomorphic copy of the circuit $C$ defined by some isomorphism $f$ mapping the variables $v_1,...,v_n$ to $v'_1,...,v'_n$. This would allow us to be able to substitute the input variables and to specify output variables for different copies of the same circuit with different inputs.

\begin{defi}
Let $C$ be a circuit and let $v_1,...,v_n \in \Gamma(C)$. Suppose there is a circuit $D$ and an isomorphism $f: \Gamma(C) \longrightarrow \Gamma(D)$ s.t. $f(x)=x$ for all $x \in {\mathcal F}(C) \setminus \{ v_1,..., v_n \}$,  $f(v_1)=v'_1,...,f(v_n)=v'_n$ and ${\mathcal E}(C) \cap {\mathcal E}(D) = \emptyset$. Then we say that $D$ is a duplicate of $C$ with $v_1/v'_1,...,v_n/v'_n$ defined by $f$ and we write $D := {\mathcal D}(C;f; v_1/v'_1,...,v_n/v'_n)$.
\end{defi}

The complexity class $TFNP$ was introduced by Megiddo and Papadimitriou in \cite{MR1107721}. This class characterizes the type of $NP$ search problems where a solution is guaranteed to exist. Let $S(X,Y)$ be a binary relation computable in polynomial time. We say that $S \in TFNP$ iff there exists a constant $c$ s.t. for any string $X$ of length $n$ there exists a string $Y$ of length $n^c$ s.t. $S(X,Y)$ holds. We consider correctness proofs of nonuniform algorithms solving the $NP$ search problem defined by the relation $S$ in propositional logic.

Let $S_n[x_1,...,x_n,y_1,...,y_{n^c};\delta]$ be the sequence of uniformly generated circuits computing the relation $S$ where $x_1,...,x_n$ and $y_1,...,y_{n^c}$ are bits of the input strings $X$ and $Y$ respectively and $\delta$ is the output s.t. the relation $S(X,Y)$ holds iff $\delta$ is true. A non-uniform algorithm for $S$ is a sequence of circuits $C_n[x_1,...,x_n;y_1,...,y_{n^c}]$ with $\Gamma(C_n) \cap \Gamma(S_n) = \{x_1,...,x_n,y_1,...,y_{n^c}\}$ where $x_1,...,x_n$ are bits of the input string $X$ and $y_1,...,y_{n^c}$ are bits of the output string $Y$. Define $Correct(C_n,\delta):=C_n \cup S_n [x_1,...,x_n;\delta] \cup \{\overline{\delta}\}$ then the nonuniform algorithm is correct if and only if the set of clauses $Correct(C_n,\delta)$ is not satisfiable for any $n$.

We show that Extended Resolution refutations of such encoded statements could be efficiently translated into resolution refutations. We shall construct from the original circuit $C_n$ a circuit $C'_n$ of size linear in the size of the given $ER$-refutation s.t. the two circuits $C_n$ and $C'_n$ are equivalent in the sense that they compute the same function yet the correctness of the new circuit $C'_n$ has efficient proofs in $R$. This is straightforwardly done by expanding the original circuit $C_n$ with the correctness definition $S_n$ along with all the extension variables as defined in the $ER$-refutation.

\begin{lem}
\label{t_c}
Let $\pi$ be an $ER$-refutation of $Correct(C_n,\delta)$ for some circuit $C_n$ then there exists an $R$-refutation of $Correct(C_n',\delta)$ of size $O(|\pi|)$ for some circuit $C'_n$ of size $O(|\pi|)$.
\end{lem}
\begin{proof}
By definition $\pi$ is an $R$-refutation of $C_n \cup S_n [x_1,...,x_n;\delta] \cup \{\overline{\delta}\} \cup D$ for some circuit $D$.

Take $C'_n := C_n[x_1,...,x_n;y_1,...,y_{n^c}] \cup {\mathcal D}(C_n \cup S_n [x_1,...,x_n;\delta] \cup D;f;x_1/x_1,...,x_n/x_n,\delta/\delta')$.

From $\pi$ we can construct an $R$-refutation of $C'_n \cup \{ \overline{\delta'}\}$ of size $O(|\pi|)$. We also know that the set of clauses $C'_n$ is not refutable since it is a circuit. Therefore the singleton clause $\{\delta'\}$ is derivable from $C'_n$ with an $R$-derivation of size $O(|\pi|)$.

Now $f \restriction \Gamma(C_n \cup S_n)$ defines an embedding from $Correct(C_n,\delta) \setminus \{ \overline{\delta} \}$ to $C'_n$. Therefore by Lemma \ref{l_eq} there is an $R$-refutation of $Correct(C_n,\delta) \cup C'_n$ of size $O(|\pi|)$. This completes the proof since $Correct(C_n,\delta) \cup C'_ n \subseteq Correct(C'_n,\delta)$.
\end{proof}

We shall briefly remark that similar results also hold in the uniform setting in bounded arithmetic.

Suppose for some polynomial time Turing machine $M$ the bounded arithmetic theory $V^1$ proves that
\begin{equation*}
V^1 \vdash \forall W \forall X \forall Y \forall Z \forall c (Comp_M(W,X,Y) \land Comp_S(Z,X,Y,c) \to c=1)
\end{equation*}
where $Comp_M(W,X,Y)$ is a formula expressing that the string $W$ encodes the transcript of computation of the Turing machine $M$ on input string $X$ and $M$ outputs the string $Y$ and $Comp_S(Z,X,Y,c)$ is a formula expressing that the string $W$ encodes the transcript of computation of the Turing machine $S$ defining the $NP$ search problem on input strings $X$ and $Y$ and output a boolean value $c$ s.t. $c=1$ if and only if the relation $S(X,Y)$ holds.

Then there exists a polynomial time Turing machine $M'$ solving the same $NP$ search problem whose correctness is provable in the theory $U^b_1$
\begin{equation*}
U^b_1\text{-IND} \vdash \forall W \forall X \forall Y \forall Z \forall c (Comp_M(W,X,Y) \land Comp_S(Z,X,Y,c) \to c=1)
\end{equation*}
where $U^b_1$-IND is a theory axiomatized by basic axioms and induction for bounded universal number quantifiers.

\section{Strength of Implicit Resolution}
\label{s_ires}
Let $n$ be the number of propositional variables and let $\Omega$ be a set of unsatisfiable clauses. We may first assume without loss of generality that every $R^*$-refutation is regular and weakening is only applied to initial clauses. Then by appropriate padding we may represent any $R^*$-refutation of $\Omega$ as a balanced decision tree of depth $n$ and size $2^{n+1}-1$.

Our encoding of circuit representation of $R^*$-refutation is as follows. Let $\beta$ be a circuit with $n+1$ inputs\footnote{To simplify definition we insert an additional input as placeholder although only $n$ inputs are actually required.} and $|n|$ outputs. The root of the balanced decision tree is computed by $\beta(0,...,0,1)$, and for each node $x=x_0,...,x_n$ in the tree, the left child is computed by $\beta(x_1,...,x_n,0)$ and the right child is computed by $\beta(x_1,...,x_n,1)$. The circuit computes the propositional variable on which the decision tree branches: the left child represents the negative literal and the right child represents the positive literal.

Notice that with such encoding, any circuit describes a correct $R^*$-refutation of some set of clauses by definition, hence one only needs to check that all the initial clauses are as prescribed.

Let $p_1,...,p_n$ be propositional variables, $\beta$ be a circuit with $n+1$ inputs and $|n|$ outputs, $L$ be a clause, we say that $L$ is an initial clause of the $R^*$-refutation described by $\beta$ if there exists a string $x_0,...,x_n$ of length $n+1$ s.t.
\begin{enumerate}[(a)]
\item
$x_0=1$
\item
$p_j \in L$ iff there exists $i < n-1$ s.t. $\beta(0,...,0,x_0,...,x_i)=j$ and $x_{i+1}=1$
\item
$\overline{p_j} \in L$ iff there exists $i < n-1$ s.t. $\beta(0,...,0,x_0,...,x_i)=j$ and $x_{i+1}=0$
\end{enumerate}
That is, from a string $1,x_1,...,x_n$ we can compute\footnote{Recall that $p^0 := \overline{p}$ and $p^1:=p$ for any propositional variable $p$.} the following initial clause $L$
\begin{equation*}
L := p^{x_1}_{\beta(0,...,0,1)} \lor p^{x_2}_{\beta(0,...,0,1,x_1)} \lor ... \lor p^{x_{n-1}}_{\beta(0,0,1,x_1,...,x_{n-2})} \lor p^{x_n}_{\beta(0,1,x_1,...,x_{n-1})}
\end{equation*}
We write ${\mathcal I}(\beta)$ to denote the set of all initial clauses of the $R^*$-refutation described by $\beta$.

Let $\Omega$ be a set of clauses in $n$ propositional variables $p_1,...,p_n$. We could encode a propositional variable as a string of length $|n|$. Similarly a literal could be encoded as a string of length $|n|+1$ with an extra bit to denote whether it is negated.

We define a circuit $\Delta(\Omega) [x_1,...,x_n, y_{1,1},...,y_{1,|n|},...,y_{n,1},...,y_{n,|n|};\delta]$ of size $O(n \cdot |n| \cdot |\Omega|)$ expressing that $\delta$ is true iff the clause $l_1 \lor ... \lor l_n$ is a weakening of some clause in $\Omega$ where $\delta$ is a new variable and $l_1,...,l_n$ are literals encoded in variables $x_1,...,x_n$, $y_{1,1},...,y_{1,|n|},...,y_{n,1},...,y_{n,|n|}$, i.e, for each $i=1,...,n$  the literal $l_i$ is encoded in the string $x_i,y_{i,1},...,y_{i,|n|}$ so that the literal $l_i$ is $p_j$ iff $y_{i,1},...,y_{i,|n|}$ codes number $j$ and it is negated iff $x_i = 0$.

To check that the clause $l_1 \lor ... \lor l_n$ is a weakening of some initial clause $L \in \Omega$ it suffices to check that for each literal $l \in L$ there exists $i=1,...,n$ s.t. $l_i=l$. Therefore by enumerating all the clauses in $\Omega$ we are able to check whether $l_1 \lor ... \lor l_n$ is a weakening of some $L \in \Omega$.

Formally $\Delta(\Omega) [x_1,...,x_n, y_{1,1},...,y_{1,|n|},...,y_{n,1},...,y_{n,|n|};\delta]$ is a set of clauses with extension variables of constant depth defined as follows:
\begin{enumerate}[$-$]
\item
$bit(m,j)$ is the $m$-th bit of the binary representation of $j$.
\item
$s_{i,j,k} \equiv x_i^{1-k} \lor \displaystyle{\bigvee_{m=1}^{|n|}y_{i,m}^{1-bit(m,j)}}$ expresses that the literal $l_i$ is not $p_j^k$.
\item
$l_{j,k} \equiv \displaystyle{\bigvee_{i=1}^n \overline{s_{i,j,k}}}$ expresses that the clause $l_1 \lor ... \lor l_n$ contains the literal $p_j^k$.
\item
$w(L) \equiv \displaystyle{\bigvee_{p_j^k \in L}\overline{l_{j,k}}}$ expresses that the clause $l_1 \lor ... \lor l_n$ is not a weakening of $L$.
\item
$\delta \equiv \displaystyle{\bigvee_{L \in \Omega} \overline{w(L)}}$ expresses that the clause $l_1 \lor ... \lor l_n$ is a weakening of some $L \in \Omega$.
\end{enumerate}
Recall that in order to compute some clause $l_1 \lor ... \lor l_n \in {\mathcal I}(\beta)$ from some string $1,z_1,...,z_n$ we need to evaluate the circuit $\beta$ $n$ times to compute all the propositional variables. We therefore define an auxiliary circuit $\Lambda[z_1,...,z_n;u_{1,0},...,u_{1,n},...,u_{n,0},...,u_{n,n}]$ of size $O(n)$ with inputs $z_1,...,z_n$ and outputs $u_{1,0},...,u_{1,n},...,u_{n,0},...,u_{n,n}$ s.t. for each $i$, $u_{i,0},...,u_{i,n}$ defines a string as input to the circuit $\beta$ in order to compute the propositional variable at depth $i$, i.e.

\begin{tabular}{c c c c c c}
$u_{1,0}=0,$ & $u_{1,1}=0,$  & \multicolumn{2}{c}{$...,$} & $u_{1,n-1}=0,$ & $u_{1,n}=1$\\
$u_{2,0}=0,$ & \multicolumn{2}{c}{$...,$} & $u_{2,n-2}=0,$ & $u_{2,n-1}=1,$ & $u_{2,n}=z_1$\\
$\vdots$ & \multicolumn{4}{c}{$\vdots$} & $\vdots$\\
$u_{n,0}=0,$ & $u_{n,1}=1,$ & \multicolumn{2}{c}{$...,$} & $u_{n,n-1}=z_{n-2},$ & $u_{n,n}=z_{n-1}$
\end{tabular}

\noindent Let $\Omega$ be a set of clauses and let $\beta[x_0,...,x_n;y_1,...,y_{|n|}]$ be a circuit with inputs $x_0,...,x_n$ and outputs $y_1,...,y_{|n|}$ we define a set of clauses ${\mathcal C}(\Omega,\beta)$ of size $O(n \cdot (|n| \cdot |\Omega| + |\beta|))$ expressing that there exists $L \in {\mathcal I}(\beta)$ s.t. $L$ is not a weakening of $L'$ for any $L' \in \Omega$.
\[\eqalign{{\mathcal C}(\Omega,\beta):= \Delta(\Omega)&[z_1,...,z_n,w_{1,1},...,w_{1,|n|},...,w_{n,1},...,w_{n,|n|};\delta]\cup \{ \overline{\delta} \} \cr
 &\cup \Lambda[z_1,...,z_n;u_{1,0},...,u_{1,n},...,u_{n,0},...,u_{n,n}] \cup \displaystyle{\bigcup_{i=1}^n \beta_i}
  }
\]
where each $\beta_i$ is a duplicate of $\beta$ computing the propositional variable at depth $i$, i.e.
\[\beta_i := {\mathcal D}(\beta;f_i;x_0/u_{i,0},...,x_n/u_{i,n},y_1/w_{i,1},...,y_{|n|}/w_{i,|n|})\]

\begin{defi}
\label{d_ires}
An implicit resolution refutation is a 4-tuple $(n, \Omega, \alpha, \beta)$ where $n$ is the number of propositional variables, $\Omega$ is the set of initial clauses, $\beta$ is a circuit with $n+1$ inputs, $|n|$ outputs and $\alpha$ is an $R$-refutation of ${\mathcal C}(\Omega,\beta)$.
\end{defi}

\begin{lem}
Implicit resolution is a Cook-Reckhow proof system.
\begin{proof}
The soundness of implicit resolution follows directly from soundness of $R$ and $R^*$.

To see that implicit resolution is complete we show that any unsatisfiable clause has an implicit resolution refutation. Let $\Omega$ be a set of unsatisfiable clauses. By completeness of $R^*$ there exists a $R^*$-refutation of $\Omega$. It is trivial to obtain a circuit $\beta$ generating such refutation by hardwiring the outputs. Now the set of clauses ${\mathcal C}(\Omega, \beta)$ is unsatisfiable since $\beta$ is indeed correct therefore it is refutable in $R$ by completeness of $R$.

In order to show that implicit resolution refutations could be verified in polynomial time we give description of a Turing machine as follows:
\begin{enumerate}[(i)]
\item\label{s_1}
decode the input in order to obtain a 4-tuple $(n, \Omega, \alpha, \beta)$ as in Definition \ref{d_ires}.
\item\label{s_2}
verify that $|\Gamma(\Omega)| \leq n$, $|{\mathcal F}(\beta)| = n + 1$ and the output variables of the circuit $\beta$ are correctly specified.
\item\label{s_3}
compute ${\mathcal C}(\Omega,\beta)$.
\item\label{s_4}
verify that $\alpha$ is an $R$-refutation of ${\mathcal C}(\Omega,\beta)$.
\end{enumerate}

It is clear that steps (\ref{s_1}) (\ref{s_2}) (\ref{s_3}) runs in polynomial time and step (\ref{s_4}) also runs in polynomial time since $R$ is also a Cook-Reckhow proof system.
\end{proof}
\end{lem}

The following is essentially a special case of Lemma 4.1 from \cite{MR2058178}, we sketch a proof here for completeness of presentation and refer the reader to \cite{MR2039508} for detailed treatment on reflection principle for resolution.

\begin{lem}
\label{l_er}
Let $\Omega$ be a set of clauses in $n$ variables and let $\pi$ be an $ER$-refutation of $\Omega$ then there exists a circuit $\beta$ of size $O(n)$ and $ER$-refutation of the set of clauses ${\mathcal C}(\Omega, \beta)$ of size $O(|\Omega|+|\pi|)$.
\begin{proof}
We construct a circuit $\beta[x_0,...,x_n;y_1,...,y_{|n|}]$ with inputs $x_0,...,x_n$ and outputs $y_1,...,y_{|n|}$. The circuit $\beta$ is canonically constructed so that it describes a decision tree that always branch on variable $p_i$ at depth $i$. Such circuit is trivially described by identifying the largest $n$ s.t. $x_i = 0$ for all $i > n$ then enumerating all the possibilities and hardwiring the outputs $y_1,...,y_{|n|}$.

It is clear that the set of clauses ${\mathcal C}(\Omega, \beta)$ is equivalent to a form of Tarski's truth definition. In other words, suppose that we have ${\mathcal C}(\Omega, \beta) = \Delta(\Omega) [x_1,...,x_n;\delta]$ for some $x_1,...,x_n$ then those $x_1,...,x_n$ are precisely the negations of the truth assignments and the entire set of clauses ${\mathcal C}(\Omega, \beta)$ is a formalized statement expressing that $\Omega$ is satisfied by the truth assignment $p_1 := \overline{x_1},...,p_n := \overline{x_n}$. To see this, notice that by definition the clause computed from the string $1,x_1,...,x_n$ is $p^{x_1}_1 \lor ... \lor p^{x_n}_n$. This clause is a weakening of some clause $L \in \Omega$ iff the truth assignment $p_1 := x_1, ..., p_n := x_n$ satisfies the clause $L$. Suppose on the contrary that $x_1,...,x_n$ satisfies the set of clauses ${\mathcal C}(\Omega, \beta)$ then the clause $p^{x_1}_1 \lor ... \lor p^{x_n}_n$ is not a weakening of any $L \in \Omega$. Hence for any $L \in \Omega$ there exists $i=1,...,n$ s.t. $p^{1-x_i}_i \in L$, that is, the truth assignment $p_1 := \overline{x_1},...,p_n := \overline{x_n}$ satisfies the set of initial clauses $\Omega$.

It is well known that such formalized truth definition could be refuted efficiently in $ER$, provided that the original set of clauses $\Omega$ has efficient $ER$-refutation.
\end{proof}
\end{lem}

\begin{thm}
\label{t_main}
Implicit resolution is p-equivalent to $ER$.
\begin{proof}
We first show that $ER$ p-simulates implicit resolution.

Let $(n, \Omega, \alpha,\beta)$ be an implicit resolution refutation we construct polynomial size $ER$-refutation of $\Omega$.

We begin by defining all the required extension variables from the set of clauses $\Omega$. Let $p_1,...,p_n$ be the propositional variables in $\Omega$. From truth assignments to $p_1,...,p_n$ we compute an initial clause defined by the string $1,z_1,...,z_n$.
\[\Lambda[z_1,...,z_n;u_{1,0},...,u_{1,n},...,u_{n,0},...,u_{n,n}] \cup \displaystyle{\bigcup_{i=1}^n \beta_i} \cup A_i [p_1,...,p_n,w_{i,1},...,w_{i,n};z_i]
\]
where each $\beta_i$ is a duplicate of $\beta$ computing the propositional variable at depth $i$, i.e.
\[\beta_i := {\mathcal D}(\beta;f_i;x_0/u_{i,0},...,x_n/u_{i,n},y_1/w_{i,1},...,y_{|n|}/w_{i,|n|})
\]
and each $A_i [p_1,...,p_n,w_{i,1},...,w_{i,n};z_i]$ is an auxiliary circuit assigning $z_i$ the truth value $\overline{p_j}$ where $j$ is the number encoded by the string $w_{i,1},...,w_{i,n}$ in binary representation.

Now the auxiliary circuits $A_i$ ensure that $z_1,...,z_n$ defines a path that falsifies $p_1,...,p_n$.
Therefore the simulation follows from the reflection principle of $R^*$ which also has polynomial size $ER$-refutation.

In this part of the proof we show that implicit resolution p-simulates $ER$.

Let $\Omega$ be the set of initial clauses. By Lemma \ref{l_er} we may assume that we are given some circuit $\beta[x_0,...,x_n;y_1,...,y_{|n|}]$ with inputs $x_0,...,x_n$ and outputs $y_1,...,y_{|n|}$ and an $ER$-refutation $\alpha$ of the set of clauses ${\mathcal C}(\Omega, \beta)$. It suffices to construct polynomial size circuit $\beta'$ and $R$-refutation $\alpha'$ of ${\mathcal C}(\Omega, \beta')$.

The $ER$-refutation $\alpha$ is in fact an $R$-refutation of ${\mathcal C}(\Omega, \beta) \cup D$ for some circuit $D$. To display the output $\delta$ let us suppose that we have ${\mathcal C}(\Omega, \beta) = \Delta(\Omega) [;\delta] \cup \{ \overline{\delta} \} \cup \Lambda[...] \cup \displaystyle{\bigcup_{i=1}^n \beta_i}$ for some $\delta$. We construct a new circuit $\beta'[x_0,...,x_n;y_1,...,y_{|n|}]$ of size $O(|\alpha|)$. Take $\beta' := \beta \cup {\mathcal D}(D \cup {\mathcal C}(\Omega,\beta)\setminus \{ \overline{\delta} \}; f;\delta/\delta')$ for some isomorphism $f$. Then a $O(|\alpha|)$ size $R$-derivation of the singleton clause $\{ \delta' \}$ from $\beta'$ could be analogously translated from $\alpha$. Now we see that the restriction $f \restriction \Gamma({\mathcal C}(\Omega,\beta) \setminus \{ \overline{\delta} \})$ defines an embedding from ${\mathcal C}(\Omega,\beta) \setminus \{ \overline{\delta} \}$ to $\beta'$. By choosing appropriate variables in ${\mathcal C}(\Omega,\beta')$ it is possible to force ${\mathcal C}(\Omega,\beta) \cup \beta' \subseteq {\mathcal C}(\Omega,\beta')$. Therefore we obtain by Lemma \ref{l_eq} an $R$-refutation of ${\mathcal C}(\Omega,\beta')$ of size $O(|\alpha|)$ as required.
\end{proof}
\end{thm}

\section{General Construction}
\label{s_gen}

Let $P$ be a Cook-Reckhow proof system defined by some Turing machine $P$ and let $\tau$ be a tautology. Let $\beta$ be a boolean circuit with $2m$ inputs. We may think of the outputs of the circuit $\beta$ as a 0-1 matrix of size $n \times n$ for some $n = 2^m$. The intended meaning is that each row of the matrix gives a snapshot of the Turing machine $P$ and the entire matrix describes a valid terminating computation of $P$ that outputs $\tau$ on some exponential size input.

Suppose for the sake of argument that $\beta$ does not describe a valid computation of $P$ that outputs $\tau$. By definition of Turing machine we would be able to identify specific local properties on which $\beta$ fails. Hence the correctness condition of the circuit $\beta$ could be expressed by formalizing the negation of the disjunction of the following statements:
\begin{enumerate}[(i)]
\item
there exists some row $j$ and some column $k$ s.t. the tape at row $j+1$ position $k$ is modified without valid tape head transition.
\item
there exists some row $j$ and some column $k$ s.t. the tape head movement at row $j+1$ from position $k$ does not conform to the transition table.
\item
there exists some row $j$ and some column $k$ s.t. row $j$ defines a terminating state and the output at position $k$ does not match the encoded tautology $\tau$.
\end{enumerate}

That is, given some $j$ and $k$ we are able to check that row $j$ column $k$ is a cell on which the circuit $\beta$ violates the definition of the Turing machine $P$ with output $\tau$, in polynomial time, provided that this is actually the case. In fact we only need to evaluate the circuit $\beta$ constantly many times.

Let ${\mathcal C}_P(\tau, \beta)$ denote the canonically generated set of clauses with limited extension expressing the conditions above. We define a general implicit proof system $[P,Q]$ based on Cook-Reckhow proof systems $P$ and $Q$. The idea is to have a $P$-proof $\alpha$ of the correctness of an exponential size $Q$-proof described by some circuit $\beta$. For convenience we shall have the correctness condition written as sets of unsatisfiable clauses instead of tautologies.

\begin{defi}[Kraj{\'{\i}}{\v{c}}ek \cite{MR2058178}]
We say that $(\alpha,\beta)$ is a $[P,Q]$-proof of the tautology $\tau$ if $\alpha$ is a $P$-refutation of the set of clauses ${\mathcal C}_Q(\tau, \beta)$.
\end{defi}

\begin{thm}
\label{t_gen}
$[EF,P] \leq_p [R,P]$ for arbitrary Cook-Reckhow proof system $P$.
\begin{proof}
It suffices to show that $[ER,P] \leq_p [R,P]$ since $EF \equiv_p ER$ implies $[EF,P] \equiv_p [ER,P]$.

Let $\tau$ be a tautology and let $(\alpha, \beta)$ be a $[ER,P]$-proof of $\tau$. We know that $\alpha$ is in fact an $R$-refutation of ${\mathcal C}_P(\tau,\beta) \cup D$ for some circuit $D$ with ${\mathcal F}(D) \subseteq \Gamma({\mathcal C}_P(\tau,\beta))$. By taking $\beta' := \beta \cup {\mathcal D}(D \cup {\mathcal C}_P(\tau,\beta) \setminus \{ \overline{\delta} \})$ where $\delta$ is the output variable asserting the correctness in ${\mathcal C}_P(\tau,\beta)$ the p-simulation follows by arguments similar to that of the proof of Theorem \ref{t_main}.
\end{proof}
\end{thm}

\begin{cor}
$iEF \leq_p [R,EF]$
\end{cor}

We define a variant of implicit resolution which is p-equivalent to $AC^0$-Frege and similar constructions also apply to subsystems of $EF$ such as Frege and $TC^0$-Frege by substituting the circuit complexity classes $NC^1$ and $TC^0$ in place of $AC^0$.

\begin{defi}
Let $\Omega$ be a set of unsatisfiable clauses in $n$ variables, we say that $(\Omega,n,\alpha,\beta)$ is an implicit $AC^0$-resolution refutation of $\Omega$ if $(\Omega,n,\alpha,\beta)$ is an implicit resolution refutation of $\Omega$ and $\beta$ is an $AC^0$ circuit.
\end{defi}

\begin{thm}
Implicit $AC^0$-resolution is p-equivalent to $AC^0$-Frege.
\begin{proof}
The proof follows directly from the arguments in Theorem \ref{t_main} with the circuit class $AC^0$ in place of $P/poly$.
\end{proof}
\end{thm}

\section*{Acknowledgement}
I would like to thank Jan Kraj{\'{\i}}{\v{c}}ek for helpful discussions and suggestions and am extremely grateful to an anonymous referee for pointing out the simplified proof of the first part of Theorem \ref{t_main} without bounded arithmetic.

%% in general the use of bibtex is encouraged

\bibliographystyle{plain}
\bibliography{wang}

\begin{thebibliography}{1}

\bibitem{MR2039508}
Albert Atserias and Mar{\'{\i}}a~Luisa Bonet.
\newblock On the automatizability of resolution and related propositional proof
  systems.
\newblock {\em Inform. and Comput.}, 189(2):182--201, 2004.

\bibitem{MR1868713}
Eli Ben-Sasson and Avi Wigderson.
\newblock Short proofs are narrow---resolution made simple.
\newblock {\em J. ACM}, 48(2):149--169, 2001.

\bibitem{MR523487}
Stephen~A. Cook and Robert~A. Reckhow.
\newblock The relative efficiency of propositional proof systems.
\newblock {\em J. Symbolic Logic}, 44(1):36--50, 1979.

\bibitem{MR821207}
Armin Haken.
\newblock The intractability of resolution.
\newblock {\em Theoret. Comput. Sci.}, 39(2-3):297--308, 1985.

\bibitem{MR2058178}
Jan Kraj{\'{\i}}{\v{c}}ek.
\newblock Implicit proofs.
\newblock {\em J. Symbolic Logic}, 69(2):387--397, 2004.

\bibitem{MR1011192}
Jan Kraj{\'{\i}}{\v{c}}ek and Pavel Pudl{\'a}k.
\newblock Propositional proof systems, the consistency of first order theories
  and the complexity of computations.
\newblock {\em J. Symbolic Logic}, 54(3):1063--1079, 1989.

\bibitem{MR1107721}
Nimrod Megiddo and Christos~H. Papadimitriou.
\newblock On total functions, existence theorems and computational complexity.
\newblock {\em Theoret. Comput. Sci.}, 81(2, Algorithms Automat. Complexity
  Games):317--324, 1991.

\bibitem{0205.00402}
G.S. Tsejtin.
\newblock {On the complexity of derivation in propositional calculus.}
\newblock {\em Semin. Math., V. A. Steklov Math. Inst., Leningrad}, 8:115--125,
  1968.

\end{thebibliography}

\end{document}